\documentclass[12pt]{article}

\usepackage{amsmath,amssymb,amsthm,verbatim}
\usepackage{url}
\usepackage{xspace}
\usepackage{MnSymbol}
\numberwithin{equation}{section}

\theoremstyle{plain}
\newtheorem{theorem}{Theorem}
\newtheorem{corollary}[theorem]{Corollary}
\newtheorem{lemma}[theorem]{Lemma}
\newtheorem{proposition}[theorem]{Proposition}
\theoremstyle{definition}
\newtheorem{remark}[theorem]{Remark}

\DeclareMathOperator{\pr}{pr}

\DeclareMathOperator{\Sym}{Sym}

\newcommand{\volform}{dx^1\wedge dx^2\wedge \dots\wedge dx^\bdim}

\newcommand{\R}[1]{\mathbb{R}^{#1}}

\newcommand{\krondel}[2]{\delta^{#1}_{#2}}
\newcommand{\He}{Helmholtz\xspace}

\newcommand{\cL}{\mathcal{L}}

\newcommand\ianhook{\mathbin{\raise1pt\hbox{\hbox{{\vbox{\hrule height.4pt
width6pt depth0pt}}}\vrule height3pt width.4pt depth0pt}\,}}
\newcommand\bdim{m}
\newcommand\bcoord[1]{x^{#1}}
\newcommand\bcoordo[1]{x_{\scriptstyle o}^{#1}}
\newcommand\connform[2]{\omega^{#1}_{#2}}
\newcommand\curvform[2]{\Omega^{#1}_{#2}}
\newcommand\curvtensor[4]{R^{#4}{}_{#1#2#3}}
\newcommand\Christoffel[3]{\Gamma^{#1}_{#2#3}}
\newcommand\Cotton[3]{C^{#1#2}_{#3}}

\newcommand\CottonEc[3]{E^{#1}_{#2}(#3)}
\newcommand{\diag}{\mathop{\mathrm{diag}}}
\newcommand\difflift{\mathfrak g}

\newcommand\dvol{\nu}

\newcommand\ELh[3]{\textup{E}^{#1#2}(#3)}
\newcommand\ELho[3]{\textup{E}^{#1#2,#3}}
\newcommand\EL[3]{\textup{E}^{#1#2}(#3)}
\newcommand\ELstand{\textup{E}^{ab}(L)}
\newcommand\ELop{\textup{E}}
\newcommand\metric{\mbox{\sffamily{g}}}
\newcommand\gc[2]{g_{#1#2}}
\newcommand\ginvc[2]{g^{#1#2}}
\newcommand\gcl[2]{g_{\lower1.3pt\hbox{$\scriptstyle #1#2$}}}
\newcommand\gco[2]{g_{{\scriptstyle o,}#1#2}}
\newcommand\gcstand{g_{ab}}
\newcommand\gcj[3]{g_{#1#2,#3}}
\newcommand\gcjl[3]{g_{\lower1.3pt\hbox{$\scriptstyle #1#2,#3$}}}
\newcommand\gccoll[1]{g^{[#1]}}
\newcommand\gccollo[1]{g_{\scriptstyle o}^{[#1]}}
\newcommand\gdet{g}

\newcommand\gind{l}
\newcommand{\gindd}{p}
\newcommand\Gbundle{\mbox{\sffamily{G}}}
\newcommand\gla[1]{{\mathfrak gl}(#1)}
\newcommand\glg[1]{Gl(#1)}
\newcommand\J[2]{J^{#1}(#2)}
\newcommand\jo{k}
\newcommand\liealg{\Gamma}
\newcommand\vliealg{\Gamma}
\newcommand\lieder[1]{\mathcal{L}_{#1}}

\newcommand\Helm{\textup{H}_\textup{T}}
\newcommand\Helmacstand{\textup{H}^{ab,cd,I}_{\textup{T}}}
\newcommand\Helmgcstand{\textup{H}^{ab,cd,I}_{\textup{T}}}
\newcommand\Helmcomp[5]{\textup{H}^{#1#2,#3#4,#5}_{\textup{T}}}
\newcommand\Helmcompm[6]{\textup{H}^{#1#2,#3#4,#5;#6}_{\textup{T}}}
\newcommand\Helmcompo[4]{\textup{H}^{#1#2,#3#4}_{\textup{T}}}
\newcommand\Helmcompom[5]{\textup{H}^{#1#2,#3#4;#5}_{\textup{T}}}
\newcommand\invpoly[2]{I^{#1}(#2)}
\newcommand\poly{\mbox{\sffamily P}}
\newcommand\polysub{\mbox{\scriptsize\sffamily P}}
\newcommand\polyc[2]{\poly^{#1}_{#2}}
\newcommand\polys{\mbox{\sffamily p}}

\newcommand\Lagr{L}
\newcommand\Lagrg{L_g}
\newcommand\matriisi{A}
\newcommand\matriisiv{B}
\newcommand\wpartial[3]{\partial^{#1#2,#3}}
\newcommand\wpartialm[2]{\partial^{#1,#2}}
\newcommand\wpartialo[2]{\partial^{#1#2}}
\newcommand\copartial[4]{\partial^{#1#2}_{#3,#4}}
\newcommand\ortha[1]{{\mathfrak so}(#1)}
\newcommand\orthg[1]{O(#1)}
\newcommand\pmi{\mathbf p}

\newcommand\qmi{\mathbf q}

\newcommand\T[2]{\textup{T}^{#1#2}}
\newcommand\Tm[3]{\textup{T}^{#1#2;#3}}
\newcommand\trace[1]{tr_{#1}}
\newcommand\Tstand[2]{\textup{T}^{#1#2}}
\newcommand\Tstandg[2]{\textup{T}_{\raise1pt\hbox{$\scriptstyle g$}}^{#1#2}}
\newcommand\Tsf{\textup{T}}

\newcommand\vf{X}%vector field
\newcommand\bvf{\xi}%vector field on the base
\newcommand\tvf{\tau}
\newcommand\liftedvf[1]{X_{#1}}%vector field
\newcommand\liftedvfevc[3]{X_{#1,\textup{ev},#2#3}}
\newcommand\vfc[1]{\xi^{#1}}
\newcommand\vfs[1]{\mathcal{X}(#1)}

\def\QED{\hskip0.1em\hfill\null\
  \null\nobreak\hfill\kern3pt\vbox{\hrule\hbox
   {\vrule\kern1pt\vbox{\kern1.7pt\hbox{$\scriptscriptstyle{QED}$}
    \kern0.2pt}\kern1pt\vrule}\hrule}}

\def\END{\hskip0.1em\hfill\null\
  \null\nobreak\hfill\kern3pt\vbox{\hrule\hbox
   {\vrule\kern1pt\vbox{\kern1.7pt\hbox{$\,\,\,\vspace{5pt}$}
    \kern0.2pt}\kern1pt\vrule}\hrule}}

\def\beginTakens{
	\begingroup                  % Start of formatting properties
	\leftskip 0.35in\rightskip 0.3in  % Wider margins for narrower text
	\it                          % Italic font
\noindent{\par}  % Make dotted line
}

\def\endTakens{
\par\noindent{\par}  % Make dotted line
	\endgroup                        % End of formatting properties
}

\newcommand{\abs}[1]{\lvert\kern.1em{#1}\rvert}
\newcommand{\norm}[1]{\lVert\kern.1em{#1}\rVert}
\newcommand*{\pd}[2]{\mathchoice{\frac{\partial #1}{\partial #2}}
  {\partial #1/\partial #2}{\partial #1/\partial #2} {\partial
  #1/\partial #2}}

  \title{Variational Principles for Natural Divergence-free Tensors in Metric Field Theories}

  \author{Ian M. Anderson
  \\
  {\footnotesize Department of Mathematics and Statistics, Utah State University,}\\
  {\footnotesize Logan, UT 84322-3900, U.S.A.}\\
  {\footnotesize \texttt{Ian.Anderson@usu.edu, http://www.math.usu.edu/$\sim$anderson/}}
  \and Juha Pohjanpelto\\
  {\footnotesize Department of Mathematics, Oregon State University,}\\
  {\footnotesize Corvallis, Oregon 97331--4605, U.S.A.}\\
  {\footnotesize \texttt{juha@math.oregonstate.edu,
      http://people.oregonstate.edu/$\sim$pohjanpp/}}
   }
  \date{}

\begin{document}

\maketitle
\begin{abstract} 
\noindent
Let $\T ab=\T ba=0$ be a system of differential equations for the components of a metric tensor on $\R\bdim$. Suppose that $\T ab$ transforms tensorially under the action of the diffeomorphism group on metrics and that the covariant divergence of $\T ab$ vanishes.  We then prove that $\T ab=\EL ab \Lagr$ is the Euler-Lagrange expression some Lagrangian density $\Lagr$  provided that $\T ab$ is of third order. Our result extends the classical works of Cartan, Weyl, Vermeil, Lovelock, and Takens on identifying field equations for the metric tensor with the symmetries and conservation laws of the Einstein equations.

\vspace{3pt}\noindent

\vspace{3pt}\noindent
\end{abstract}

% ----------------------------------------------------------------------
\section{Introduction and main results}

The time-honored Noether's theorems \cite{Olv93} establish a correspondence between symmetries and conservation laws for the Euler-Lagrange equations of a classical variational problem. Noether's first theorem states that every infinitesimal symmetry of the variational problem determines a differential conservation law for the Euler-Lagrange equations, and conversely that, under certain mild non-degeneracy conditions, conservation laws can be associated with symmetries of the variational problem. Noether's second theorem, in turn, asserts that infinite dimensional symmetry pseudogroups of the variational problem involving  arbitrary functions of all the independent variables correspond to differential identities among the equations. 

In 1977, F.~Takens \cite{Tak77} considered the following novel and very distinct aspect of the interplay between symmetries, conservation laws and variational principles for systems of differential equations:
\vspace{5pt}

\beginTakens

\noindent Let $\liealg$ be a Lie algebra of vector fields
defined on the space of independent and dependent variables, and suppose that a
system of differential equations is invariant under $\liealg$ and that each
element in $\liealg$ generates a conservation law for the system.  Does it then
follow that the system arises from a variational principle?
% i.e., that it is
%the Euler-Lagrange system associated with some Lagrangian function?
\endTakens
\vskip5pt
In his original
paper Takens studied the question for second order scalar equations, systems
of linear equations, and equations arising in metric field theories. Takens' results
for second order scalar equations and for systems of linear equations were
subsequently generalized by Anderson and Pohjanpelto \cite{AnPo94},
\cite{AnPo95}, \cite{Poh}.  We refer to \cite{AnPo94} in particular for more
background material and motivation on Takens' problem.

In addition to the papers listed above, there is an extensive literature dealing with the existence
of variational principles for systems of differential equations admitting a Lie algebra of
symmetries and the corresponding conservation laws within the context of Noether's second theorem, that is, where the symmetry group is the infinite dimensional
group of coordinate transformations of the underlying manifold and the
conservation laws express the vanishing of the covariant divergence (or some
variant of it) of the field equations.  This work is largely directed towards the axiomatic characterization of the Einstein equations. 
The original classification results of Cartan
\cite{Ca22}, Vermeil \cite{Ve17}, and Weyl \cite{We22} establish that second order
quasi-linear field equations for the metric tensor possessing the symmetries
and conservation laws of the Einstein equations necessarily arise from a
variational principle. These results were later extended to general, fully nonlinear second order
systems for the metric tensor and to third order systems in the 3-dimensional case by Lovelock, \cite{Lo70, Lo71, Lo76}, again by a detailed classification of all equations with the required properties. A direct proof for second order systems based on the analysis of the \He conditions \cite{And}, i.e., the integrability conditions for the existence of a variational principle, can be found in \cite{Tak77}. Lovelock's results were later generalized to metric-scalar
\cite{HorLov72}, \cite{Hor74b}, metric-vector \cite{Lo74}, and metric-bivector \cite{Lo77}
theories. Anderson \cite{And78} subsequently provides a general characterization of second order divergence free systems for the metric tensor and an auxiliary independent tensor field, subsuming in part the above-mentioned works on combined metric field theories.

In \cite{Hor75}, Horndeski  attempts  to extend Lovelock's work to general third order equations with the symmetries and conservation laws of the Einstein equations. However, the treatment for the existence of a variational principle falls short of a comprehensive result due to a restrictive extraneous assumption that the zeroth order Helmoltz conditions for the system be invariant under the action of the diffeomorphism group. 

The next important step extending Lovelock's work was the introduction of the generalized Cotton tensors in \cite{Ald1980}.
These tensors and their construction is placed into the proper differential geometric context by Anderson \cite{And84} as the Euler-Lagrange expressions of Lagrangians derived from the Chern-Simons forms of Riemannian geometry. As is discovered in \cite{And84}, these analogs of the classical Cotton tensors of 3-dimensional conformal geometry play a key role in the equivariant inverse problem of the calculus of variations for metrics.

The aim of this paper is to extend the results of Cartan, Vermeil, Weyl, and Lovelock to general third order field equations for the metric tensor with the symmetries and conservation laws of the Einstein equations. As is intimated by the intricacy of the constructions in \cite{Ald1980}, \cite{And84}, a direct classification of third order equations along the lines of the original works on the subject would be a formidable undertaking. 
However, as is well known, the existence of a local Lagrangian for a system of differential equations is ensured by the vanishing of the classical \He conditions for the equations, and our problem is
rendered tractable by an analysis of these conditions for systems sharing the properties of the Einstein equations.

To describe our results more precisely, write $(\bcoord a)$ for the coordinates on $\R{\bdim}$. A metric
$$\metric{}{}=\gc a b d\bcoord a\otimes d\bcoord b$$
on $\R\bdim$ is a symmetric type $(0,2)$ tensor field with $\gdet = \det(\gc a b)\neq0$, where the $\gc a b$ stand for the components of $\metric$. In this paper we consider metrics of fixed but arbitrary signature. The action of 
the Lie algebra $\vfs{\R\bdim}$ of vector fields 
%the infinitesimal diffeomorphism group group $\vfs{\R\bdim}$ of $\R\bdim$ 
on metric tensors via pull-back gives rise to the infinitesimal transformation group
  \begin{equation}\label{eq:difflift}
  \difflift = \left\{\liftedvf\bvf=\vfc i\pd{}{\bcoord i} -2\vfc c_{,(a}\gc{\lower1.6pt\hbox{$\scriptstyle b)$}}{\lower1.6pt\hbox{$\scriptstyle c$}}
      \pd{}{\gc a b} \mid \vfc a\in
    \vfs{\R\bdim}\right\}
  \end{equation}
on the coordinate space $\Gbundle=\{(\bcoord i,\gc bc)\}$.

The metric tensor $\metric$ is subject to a system of $k$th
order partial differential equations
\[
\Tstand a b=\Tstand {(a}{ b)}(\bcoord i,\gc c d,\gcj c d{i_1}, \gcj c d{i_1i_2},\ldots,
\gcj c d{i_1i_2\cdots i_k})=0,\qquad a,b=1,\dots, \bdim,
\]
where $\gcj c d{i_1i_2\cdots i_l}$ denotes the derivative of $\gc c d$
with respect to the independent variables $\bcoord {i_1}$,$\bcoord {i_2}$,\dots, $\bcoord {i_l}$. The
operator $\Tstand a b$ is locally variational if it can be written in some neighborhood of each point of its domain as the Euler--Lagrange expression
\begin{equation}\label{eq:EulerLgrangeT}
\Tstand a b=\EL a b{\Lagr}= \pd{\Lagr}{\gc a b} -D_{i_1}\left(\pd{\Lagr}{\gcj a b{i_1}}\right)+
D_{i_1}D_{i_2}\left(\pd{\Lagr}{\gcj a b{i_1i_2}}\right)-\cdots
\end{equation}
of some locally defined Lagrangian 
\[
\Lagr = \Lagr(\bcoord c,\gc c d,\gcj c d{i_1}, \gcj c d{i_1i_2},\ldots,\gcj c d{i_1i_2\cdots i_l})
\] 
depending on the components of the metric tensor and their derivatives. Here $D_i$ denotes the
standard coordinate total derivative operator.

If the Lagrangian $\Lagr$ transforms as a scalar density, that is,
\begin{equation}\label{E:LScalarDensity}
\lieder{\pr\liftedvf{\vfc{}}}\Lagr={\pr\liftedvf{\vfc{}}}(\Lagr) + \mathrm{div}\,\vfc{}\,\Lagr=0,\qquad
\text{for all $\liftedvf{\vfc{}}\in\difflift$},
\end{equation}
 or, equivalently,
the density $\Lagrg=g^{-1/2}\Lagr$ is invariant under the prolonged action of $\difflift$, then the Euler-Lagrange expressions $\Tstand ab=\EL a b \Lagr$ constitute a tensor density, whereby
\begin{equation}\label{eq:TensorDensity}
\lieder{\pr\liftedvf{\vfc{}}}\Tstand ab={\pr\liftedvf{\vfc{}}}(\Tstand ab) - 2\vfc{(a}_{,c}\Tstand{b)}c+\mathrm{div}\,\vfc{}\,\Tstand ab=0,\qquad\text{for all $\liftedvf{\vfc{}}\in\difflift$}.
\end{equation}
Here $\lieder{}$ denotes the standard Lie derivative operator. Lagrangians $\Lagr$ and differential operators $\Tstand ab$ satisfying \eqref{E:LScalarDensity} or \eqref{eq:TensorDensity} are also known as \emph{natural} tensor densities.

In addition, in light of the diffeomorphism invariance \eqref{eq:TensorDensity}, Noether's second theorem implies that the components $\Tstand a b=\EL ab\Lagr$ are divergence-free,
\begin{equation}\label{eq:CovarDiv}
D_b \Tstand a b + \Christoffel a b c \Tstand b c=0,
\end{equation}
where the $\Christoffel a b c$ denote the standard Christoffel symbols of the metric $\metric$.

In the present paper we investigate a partial converse to Noether's second theorem for third order operators for metrics, that is, whether a symmetric, type $(2,0)$ differential operator
$\Tstand ab =\Tstand{(a}{b)}(\bcoord i, \gc c d,\gcj c d{i_1}, \gcj c d{i_1i_2},\gcj c d{i_1i_2i_3})$ satisfying the invariance condition \eqref{eq:TensorDensity} and subject to the differential constraints \eqref{eq:CovarDiv} coincides with the Euler-Lagrange expression of some Lagrangian $\Lagr$.

Our main result is the following.
\smallskip
\begin{theorem}\label{th:main1}
  Suppose that a third order differential operator
\[\Tstand a b=\Tstand {(a}{ b)}(\bcoord i,\gc c d,\gcj c d{i_1}, \gcj c d{i_1i_2},\gcj c d{i_1i_2i_3}),\qquad a,b=1,\dots, \bdim,
\]
admits the symmetries
\eqref{eq:TensorDensity} and satisfies the differential constraints \eqref{eq:CovarDiv}. Then $\Tstand a b$ is locally variational. Moreover, suppose that $\T ab$ is everywhere smooth. Then it can be written as
\begin{equation}\label{eq:InverseProblem}
\T ab = \begin{cases} \EL ab\Lagr, \quad&\text{if $\bdim \equiv 0,1,2 \mod 4$,}\\
\EL ab\Lagr + \Cotton ab \polysub,\quad&\text{if $\bdim \equiv 3 \mod 4$,}
\end{cases}
\end{equation}
where the Lagrangian $\Lagr$ is a scalar density satisfying \eqref{E:LScalarDensity} and $\Cotton ab \polysub$ is the generalized Cotton tensor associated with an invariant polynomial $\poly$ on  $\ortha{p,q}$, $p+q=\bdim$, of degree $(\bdim+1)/2$.
\end{theorem}

The generalized Cotton tensors $\Cotton ab \polysub$ are, therefore, locally variational, but, as is proved in \cite{And84}, they can not be written as the Euler-Lagrange expressions of a natural Lagrangian.
In the physically most relevant situation with $\bdim =4$, Theorem \ref{th:main1} asserts that if a third order, natural system of differential equations 
\[\Tstand {a}{ b}(\gc c d,\gcj c d{i_1}, \gcj c d{i_1i_2},\gcj c d{i_1i_2i_3})=0,\qquad\text{where $\Tstand a b=\Tstand {(a}{ b)}$},\]
for the components $\gc cd$ of the metric tensor 
% admits the symmetries \eqref{eq:TensorDensity} and 
is divergence free \eqref{eq:CovarDiv}, then there is a natural Lagrangian $\Lagr$ so that $\Tstand ab = \EL ab \Lagr$.
  
\smallskip
This paper is organized as follows. After covering some preliminary material
relevant to the problem at hand in section \ref{S:prelims}, we analyze in section
\ref{sec:symcl} the relationship between symmetries \eqref{eq:TensorDensity} and the 
conservation law \eqref{eq:CovarDiv} for metric field equations. In particular, we show
that any natural differential operator $\Tstand a b$ admitting translational conservation laws is necessarily divergence free. This
interesting though elementary fact does not seem to have been previously noted
in the literature.  Then in section \ref{sec:proof} we present the proof of Theorem
\ref{th:main1} and, finally, in section \ref{sec:discussion} we discuss some open problems and generalizations of the work at hand. 

\section{Preliminaries}\label{S:prelims}

In this section we collect together some basic definitions and results from the formal calculus of variations on jet spaces germane to the problem at hand. For more details and proofs we refer, e.g., to~\cite{And,Olv93}.

Let $\Gbundle\to\R\bdim$ be the trivial bundle of metrics, that is, of non-degenerate symmetric bilinear forms on $\R\bdim$ with fixed signature.  Denote the coordinates of $\R\bdim$ by $(\bcoord 1,\bcoord 2,\dots,\bcoord\bdim)$. Then the components $\gc a b=\gc{(a}{b)}$ of a metric $\metric$ are determined by $\metric = \gc a b d\bcoord a\otimes d\bcoord b$, so that, as a coordinate bundle,
\[
\Gbundle=\{(\bcoord i,\gc a b)\}\to \{(\bcoord i)\},\qquad \text{where $a\leq b$}.
\]
We denote the bundle of order $\jo$ jets, $0\leq \jo\leq\infty$, of local sections of $\Gbundle$ by $\J\jo \Gbundle$;
in the induced coordinates
\begin{equation}\label{eq:YMcoord}
  \J\jo \Gbundle = \{(x^i,\gc a b,\gcj a b{i_1},
  \gcj a b{i_1i_2},\ldots,\gcj a b{i_1i_2\cdots i_\gind},\ldots,\gcj a b{i_1i_2\cdots i_\jo})\},
\end{equation}
where $\gcj a b{i_1i_2\cdots i_\gind}$, $a\leq b$, stands for the $\gind$-th order derivative
variables. For notational convenience we let $\gcj a b{i_1i_2\cdots i_\gind}=\gcj b a{i_1i_2\cdots i_\gind}$, when $a>b$. We also  use $\gccoll k$ to collectively designate all the variables $\gcj a b{i_1\cdots i_p}$, $p=0,\dots,\jo$, up to order $\jo$.

Let $I=(i_1,i_2,\dots,i_\gind)$, $1\leq i_\gindd\leq\bdim$, denote an unordered multi-index
of length $|I|=\gind$.  Define partial derivative operators $\wpartial a b I$, $1\leq a,b\leq\bdim$, $|I|\geq0$,
by
\begin{equation}\label{eq:wpd}
  \wpartial a b I  \gcj c d J = \begin{cases}
    \krondel{\hphantom{(}a}{(c} \krondel{b}{d)}
    \krondel {(i_1} {\hphantom{(}j_1}\krondel {i_2} {j_2}\cdots\krondel {i_k)} {j_k},
    \qquad&\text{if $|I|=|J|$},\\
    0,\qquad&\text{if $|I|\neq |J|$},
  \end{cases}
\end{equation}
where round brackets indicate symmetrization in the enclosed indices. Then, for example, the
standard coordinate total derivative operators $D_i$ on $\J\infty\Gbundle$ are given in terms of the differential operators \eqref{eq:wpd} by
\begin{equation}
  \label{eq:totalder}
  D_i = \pd{}{\bcoord i}+
  \sum_{|I|\geq0}
  \gcj a b{Ii}\wpartial a b I,
  \qquad i=1,2,\dots,\bdim.
\end{equation}
The expression \eqref{eq:totalder} leads to the commutation formula
\begin{equation}\label{eq:totaldercomm}
[\wpartial ab I,D_j] = \wpartial ab{(i_1\cdots i_{k-1}}\delta^{i_k)}_j.
\end{equation}
We will employ the standard Einstein summation convention in what follows.

The flow of a vector field
\begin{equation}\label{eq:vectorfield}
  \vf= P^i(x^j,\gc c d)\pd{}{x^i}
  +Q_{ab}(x^j,\gc c d)\wpartialo a b
\end{equation}
on $\Gbundle$ induces a transformation on the space of sections of $\Gbundle$,
and, hence, by differentiation, it generates a local 1-parameter transformation group acting on
$\J k\Gbundle$, $k\geq0$. The associated infinitesimal generator is called the
$k$-{th} order \emph{prolongation} of $X$ and is denoted by $\pr^k X$. The components of $\pr^k X$
are given by the usual prolongation formula
\begin{equation}\label{eq:prolongation}
  \pr^k X = P^i D_i+\sum_{|I|\leq k}
  D_I (X_{\text{ev},ab})
  \wpartial a b I,
\end{equation}
where the $X_{\text{ev},ab}$ denote the components of the
\emph{evolutionary form}
\begin{equation*}
  X_{\text{ev}} = (Q_{ab}-P^c\gcj a b c)
  \partial^{ab}
\end{equation*}
of $X$ and where, given a multi-index $I=(i_1,\dots,i_k)$, we use the abbreviated notation $D_I = D_{i_1}\cdots D_{i_k}$ . We will also write $\pr^\infty X = \pr X$.  The vector field
\eqref{eq:vectorfield} is called \emph{projectable} if the coefficients
$P^a=P^a(x^i)$ are functions of the independent variables $x^i$ only. In particular, the infinitesimal generators of the action by the lifted diffeomorphism group \eqref{eq:difflift} form a Lie algebra $\mathfrak g$ of projectable vector fields on $\Gbundle$ with 
\[
\liftedvfevc{\vfc{}}ab=-2\vfc c_{,(a}\gc{\lower1.3pt\hbox{$\scriptstyle b)$}}{\lower1.3pt\hbox{$\scriptstyle c$}}
-\vfc c\gcj abc.
\]

We associate to a given differential operator $\Tstand a b=\Tstand a b (x^i,\gccoll k)$ the \emph{source form}
\begin{equation}\label{eq:sourceform}
  \Tsf = \Tstand a b d\gc a b \wedge\dvol,
\end{equation}
where $\dvol = \volform$ is the standard volume form on $\R\bdim$.  The source form $\Tsf$ is called \emph{natural} if $\Tsf$ is invariant under the prolonged action of the diffeomorphism group $\difflift$, that is,
\[
\lieder{\pr\liftedvf{\vfc{}}}\Tsf =0,\qquad\text{for all $\vfc{}\in\vfs{\R\bdim}$,}
\]
where $\liftedvf{\vfc{}}$ is defined in \eqref{eq:difflift}. As is easy to verify, a source form $\Tsf = \Tstand a b d\gc a b \wedge\dvol$ is natural precisely when the components $\Tstand ab$ form a tensor density as in \eqref{eq:TensorDensity}.
 A vector field $X$ on $\Gbundle$ \emph{generates a conservation law} for $T$ if there are
differential functions $t^i= t^i(x^j,\gccoll l)$, $i=1,\dots,\bdim$, so that
\begin{equation}
X_{\text{ev,}}{}_{ab} \Tstand a b
  % = \operatorname{Div} t
  = D_i t^i.
\end{equation}

The source form $\Tsf$ is said to be derivable from a \emph{variational principle} if there is a
\emph{Lagrangian function} $L= L(x^i,\gccoll l)$ such that $\Tsf$ is the \emph{Euler-Lagrange
expression} of $L$, {i.e.},
\begin{equation}\label{eq:EulerLagrange}
\Tstand a b  = \ELh a b L =
\sum_{|I|\geq 0}
(-D)_I (\wpartial a b I L).
\end{equation}
It will be convenient to call
\[\lambda = L(x^a,\gccoll l)\,\dvol\]
a \emph{Lagrangian $n$-form} and
\[
\ELop(\lambda) = \ELstand\,d\gc a b\wedge\dvol
\]
the \emph{Euler-Lagrange
form} associated with $\lambda$.  As is well known \cite{And}, the Euler-Lagrange operator
commutes with the prolonged action of projectable transformations on
$\Gbundle$; infinitesimally,
\begin{equation}\label{eq:EulerEquivariance}
  \ELop(\cL_{\pr X}\lambda) = \cL_{\pr X}\ELop(\lambda),
\end{equation}
for every projectable vector field $X$ and Lagrangian form $\lambda$, where
$\cL$ denotes the standard Lie derivative operator.

The \He operator $\Helm$ acts on
evolutionary vector fields $Y = Y_{ab}\wpartialo ab$ on $\Gbundle$ by
\begin{equation}\label{eq:Helmdefn}
  \Helm(Y) = \mathcal L_{\pr Y}\Tsf
  -\ELop(Y\ianhook\mkern2mu \Tsf),
\end{equation}
cf.~\cite{AnPo95b}. If we write
\begin{equation}\label{eq:Helmcompdefn}
  \Helm(Y) =
  \sum_{|I|\geq0}
  (D_IY_{cd})\Helmgcstand d\gcstand\wedge\dvol,
\end{equation}
then the components $\Helmgcstand$ of $\Helm$ are explicitly given by
\begin{equation}\label{eq:HelmComp}
  \Helmgcstand =
  \wpartial c d I \Tstand a b -
  (-1)^{|I|} \ELho a b I(\Tstand cd),\qquad |I|\geq 0.
\end{equation}
Here the $\ELho a b I$ denote the higher Euler-Lagrange operators \cite{And} acting on
a differential function $F$ defined on some $\J k\Gbundle$ by
\[
\ELho a b I(F) = \sum_{|J|\geq0} \tbinom {|I|+|J|}{|I|}
(-D)_J(\wpartial ab{IJ} F), \qquad |I|\geq 0.
\]
Note that if $\Tsf = \Tstand a b (x^i,\gccoll k) d\gcstand \wedge\dvol$ is of order
$k$, then $\Helmgcstand=0$ for $|I|>k$ and for $|I|=0,\dots, k$, the
components $\Helmgcstand$ are of order at most $2k-|I|$.

It is not difficult to see that a source form $\Tsf = \ELop(L)$ deriving from a variational
principle satisfies the \He conditions $\Helm\equiv0$, or, in components,
$\Helmgcstand=0$.  Conversely, one can show \cite{And} that if the \He
conditions $\Helm\equiv0$ are satisfied, then, at least locally, that is, in some neighborhood of each point in its domain in $\J k\Gbundle$, the source form $\Tsf$
can be written as the Euler-Lagrange expression of some Lagrangian $L$.
Accordingly, we will call a source form satisfying the \He conditions
\emph{locally variational}.
\smallskip
\begin{proposition}\label{pro:Lieder}
  Suppose that $X = P^i\pd{}{x^i}+Q_{ab}\wpartialo ab$ is a
  projectable vector field on $\Gbundle$ and that a source form $\Tsf$ is
  invariant under the prolongation $\pr X$ of $X$.  Then the components
  $\Helmgcstand$ of the \He operator $\Helm$ associated with $\Tsf$ satisfy
  the invariance conditions
  \begin{equation}\label{eq:HelmCompInv}
    {\pr X}(\Helmgcstand)
    +\sum_{|J|\geq|I|}\Helmcomp ab ef J\wpartial cd I Q_{ef,J}
    +\Helmcomp ef cd I
    \wpartialo ab Q_{ef}
    +\pd{P^j}{x^j}\Helmacstand=0,
  \end{equation}
  where $Q_{cd,J}$ denotes the $\gcj cd J$-component of $\pr X$.
\end{proposition}
\smallskip

\begin{proof}
  We first compute
  \begin{equation}\label{eq:XderHelm}\begin{split}
      \cL_{\pr X}\left(\Helm(Y)\right) &= \cL_{\pr X}(\cL_{\pr Y}\Tsf)
      -\cL_{\pr X}\ELop(Y\ianhook \Tsf)\\
      &= \cL_{\pr [X,Y]}\Tsf-\ELop((\cL_{\pr X}Y)\ianhook \Tsf) =\Helm([X,Y]),
    \end{split}
  \end{equation}
  where we used the invariance of $\Tsf$ and the equivariance of the Euler-Lagrange operator under the prolonged action of projectable
transformations.  Next write $Y = Y_{cd}\wpartialo cd{}$.  Then, on
  account of \eqref{eq:Helmcompdefn},
  \begin{equation}\label{eq:XderHelmcomp1}\begin{split}
      \cL_{\pr X}\left(\Helm(Y)\right) &= \sum_{|I|\geq0}\bigg({\pr
        X}(\Helmgcstand) D_IY_{cd}
      +\Helmgcstand {\pr X} (D_I Y_{cd})\\
      &\quad+\Helmcomp efcd I (D_IY_{cd}) \wpartialo ab Q_{ef}
      +\pd{P^j}{x^j}\;\Helmgcstand D_I Y_{cd}\bigg) d\gcstand\wedge\dvol.
    \end{split}\end{equation}
The identity
\[
[\pr X,\pr Y] = \pr [X,Y]
\]
yields
\begin{equation}\label{eq:XdertotderY}
\begin{split}
{\pr X} (D_I Y_{cd})&= D_I([X,Y]_{cd})+\pr Y(Q_{cd,I})\\
&=D_I([X,Y]_{cd}) +\sum_{|J|\geq0}(D_J Y_{ef})\wpartial ef J Q_{cd,I}.
\end{split}
\end{equation}
  Now by virtue of
  \eqref{eq:XdertotderY}, equation \eqref{eq:XderHelmcomp1} becomes
  \begin{equation}\label{eq:XderHelmcomp2}
    \begin{split}
      \cL_{\pr X}\left(\Helm(Y)\right) &= \sum_{|I|\geq0}\bigg({\pr
        X}(\Helmgcstand) D_IY_{cd}
      +\Helmgcstand D_I([X,Y]_{cd})\\
      &\qquad+\sum_{|J|\geq0}\Helmgcstand(D_J Y_{ef})\wpartial ef J Q_{cd,I} +\Helmcomp efcd I (D_IY_{cd}) \wpartialo ab Q_{ef}\\
      &\qquad\qquad+\pd{P^j}{x^j}\;\Helmacstand D_I Y_{cd}\bigg)
      d\gcstand\wedge\dvol.
    \end{split}
  \end{equation}
  A comparison of \eqref{eq:XderHelm} with \eqref{eq:XderHelmcomp2}
  yields equation \eqref{eq:HelmCompInv}, as required.
\end{proof}

\begin{proposition}\label{P:HelmIntCond} Let $\Tsf$ be a third order source form. Then the components $\Helmcomp abcd{}$, $\Helmcomp abcd i$, $\Helmcomp abcd{ij}$, $\Helmcomp abcd{ijk}$ of the \He operator associated with $\Tsf$ satisfy the integrability conditions 
\begin{subequations}\label{eq:HelmIntCond}
\begin{align}
\Helmcompo abcd{}+\Helmcompo cdab{} &=D_i\Helmcomp abcd i-D_iD_j\Helmcomp abcd {ij}+D_iD_jD_k\Helmcomp abcd{ijk},\label{eq:HelmIntCond0}\\
\Helmcomp abcd i-\Helmcomp cdab i &=2D_j\Helmcomp abcd{ij}-3D_jD_k\Helmcomp abcd{ijk},\label{eq:HelmIntCond1}\\
\Helmcomp abcd{ij} + \Helmcomp cdab{ij}&=3D_k\Helmcomp abcd{ijk},\label{eq:HelmIntCond2}\\
\Helmcomp abcd{ijk} - \Helmcomp cdab{ijk}&=0\label{eq:HelmIntCond3}.
\end{align}
\end{subequations}
\end{proposition}

\begin{proof} Equations \eqref{eq:HelmIntCond} can be derived by direct computation using the expressions \eqref{eq:HelmComp} for $\Helmgcstand$. Since $\Tsf$ of order 3, these yield
\begin{equation}\label{eq:ThirdOrderHe}
\Helmcomp abcd{ijk} = \wpartial cd{ijk}\T ab + \wpartial ab{ijk}\T cd,
\end{equation}
which immediately implies equation \eqref{eq:HelmIntCond3}. Likewise, by \eqref{eq:HelmComp},
\[
\Helmcomp abcd{ij}+\Helmcomp cdab{ij}= 3D_k(\wpartial ab{ijk}\T cd+\wpartial cd{ijk}\T ab)=3D_k\Helmcomp abcd{ijk},
\]
so that \eqref{eq:HelmIntCond2} holds true. Next we compute
\[
\begin{split}
\Helmcomp abcd{i}&-\Helmcomp cdab{i}\\
&= 2D_j(\wpartial cd{ij}\T ab-\wpartial ab{ij}\T cd) +3D_jD_k(\wpartial ab{ijk}\T cd-\wpartial cd{ijk}\T ab)\\
&=2D_j(\wpartial cd{ij}\T ab-\wpartial ab{ij}\T cd+3D_k\wpartial ab{ijk}\T cd)
-3D_jD_k(\wpartial ab{ijk}\T cd+\wpartial cd{ijk}\T ab)\\
&=2D_j\Helmcomp abcd{ij}-3D_jD_k\Helmcomp abcd{ijk},
\end{split}
\]
which yields \eqref{eq:HelmIntCond1}. Finally,
\[
\begin{split}
\Helmcompo abcd&+\Helmcompo cdab = D_i(\wpartial cd i\T ab+\wpartial ab i\T cd)\\
&\quad- D_iD_j(\wpartial cd{ij}\T ab+\wpartial ab{ij}\T cd)+D_iD_jD_k(\wpartial cd{ijk}\T ab+\wpartial ab{ijk}\T cd)\\
&=D_i(\wpartial cd i\T ab+\wpartial ab i\T cd-2D_j\wpartial ab{ij}\T cd+3D_jD_k\wpartial ab{ijk}\T cd)\\
&\quad - D_iD_j(\wpartial cd{ij}\T ab-\wpartial ab{ij}\T cd +3D_k\wpartial ab{ijk}\T cd)
+D_iD_jD_k(\wpartial cd{ijk}\T ab +\wpartial ab{ijk}\T cd)\\
&=D_i\Helmcomp abcd i-D_iD_j\Helmcomp abcd{ij}+D_iD_jD_k\Helmcomp abcd{ijk},
\end{split}
\]
which proves the first identity \eqref{eq:HelmIntCond0}.
\end{proof}

\begin{remark}
The \He operator $\Helm$ can also be characterized in terms of the differential $\delta_V$
of the variational bicomplex \cite{And} as $\Helm = \delta_V\Tsf$, whereby the foregoing conditions ensue from the general identity $\delta_V^2=0$. 
%$\delta_V\Helm =\delta_V^2\Tsf=0$. 
\end{remark}

The following Lie derivative formula, as established in \cite{And,AnPo94}, is central in the proof of our main Theorem.
\smallskip
\begin{proposition} Let $\Tsf$ be a source form and $X$ a projectable vector field
  on $\Gbundle$.  Then
  \begin{equation}\label{eq:Liederformula}
    \cL_{\pr X}\Tsf =
    \ELop(X_{{\textup{ev}}}\ianhook \Tsf) + \Helm(X_{\textup{ev}}).
  \end{equation}
\end{proposition}
An extension of the Lie derivative formula \eqref{eq:Liederformula} to
non-projectable, generalized vector fields can be found in \cite{And}.  If $\Tsf$ is a locally variational source form, then equation
\eqref{eq:Liederformula} reduces to
\[\cL_{\pr X}\Tsf =\ELop(X_{\text{ev}}\ianhook \Tsf).\]
Now $$\cL_{\pr X}\Tsf = 0,$$ if the source form is
invariant under $X$, that is, if $X$ is a \emph{distinguished symmetry}
of the system $\T ab=0$, while $$\ELop(X_{\text{ev}}\ianhook \Tsf)=0,$$ provided that $X$ generates a conservation law for $\Tsf$; see \cite{And}.  Thus equation
\eqref{eq:Liederformula} furnishes a version of the classical Noether's theorem for
projectable vector fields expressed directly in terms of the system of
differential equations without explicit reference to a Lagrangian.

On the other hand, in the situation of Takens' problem, each vector field $X$ belonging to
the Lie algebra $\vliealg$ is, by prescription, a symmetry of the source form $\Tsf$ and generates a
conservation law for $\Tsf$. These requirements lead to the conditions
\[
\Helm(X_{\text{ev}})=0,\qquad\text{ for
all $X\in\vliealg$},
\]
on the \He operator of $\Tsf$.  The primary objective in the
analysis of Takens' problem is to identify, on mathematical or physical grounds,
interesting classes $\mathcal T$ of source forms ({i.e.,}~differential
equations) and symmetry algebras $\vliealg$ of vector fields so that one will
be able to classify all $\vliealg$-invariant \He operators $\Helm$ with $\Tsf\in\mathcal T$ satisfying the conditions
$\Helm(X_{\text{ev}})=0$, $X\in\vliealg$.

We conclude this section by briefly recalling the construction of the generalized Cotton tensors $\Cotton {}{}P$. For more details and proofs we refer to \cite{And84}. 

Let $\connform ij=\Christoffel ijk dx^k$, $i,j=1,\dots,\bdim$, denote the connection forms of the Riemannian connection associated with a metric $\metric$, and, as usual, write 
\[
\curvform ij = d\connform ij +\connform ik\wedge\connform kj
\]
for the associated curvature form. In terms of the components of the Riemannian curvature tensor,
\[
\curvform ij = \dfrac12\curvtensor jkli dx^k\wedge dx^l.
\]

Next let 
\[
\trace r\mkern1mu\colon \gla\bdim\to\R{},\quad \trace r(\matriisi)=\text{trace}(\matriisi^r),\quad r=1,2,\dots,
\]
denote the elementary polynomials on the Lie algebra $\gla\bdim$ invariant under the adjoint action of the general linear group $\glg\bdim$. As is well known \cite{Spi75}, any $\orthg{p,q}$, $\bdim=p+q$, invariant polynomial $\poly\in\invpoly{2n}{\ortha{p,q}}$ of homogeneous degree $2n$ on the Lie algebra $\ortha{p,q}$ of the orthogonal group $\orthg{p,q}$ can be expressed as $\poly=\polys(\trace 2,\trace 4,\dots,\trace{2s})$ for a uniquely determined polynomial $\polys\colon\R s\to \R{}$.

Next assume that $m=4n-1$. We associate to a polynomial $\poly=\polys(\trace 2,\trace 4,\dots,\trace{2s})\in\invpoly{2n}{\ortha {p,q}}$ the second order differential functions $\CottonEc{abc}{\polysub}\metric$ by 
\[
\CottonEc{abc}{\polysub}\metric\dvol = \ginvc ad\polyc cd(\curvform{}{})\wedge d\bcoord b,
\]
where the $\polyc cd$ are the components of the derivative of $\poly$, that is,
\[
\frac{d}{dt}\poly(\matriisi+t\matriisiv)\vert_{t=0} = \polyc cd(\matriisi)\matriisiv^d_c.
\]
Now the generalized Cotton tensors are defined by
\[
\Cotton {}{}\polysub(\metric) = \Cotton ab\polysub\mkern1mu d\gc ab\wedge\dvol,
\]
where the components $\Cotton ab\polysub$ are, by definition, given by the covariant derivatives
\[
\Cotton ab\polysub= \frac12\nabla_c(\CottonEc{acb}\polysub\metric+\CottonEc{bca}\polysub\metric).
\]
One can show that $\Cotton {}{}\polysub(\metric)$ is a third order, locally variational natural source form \cite{And84}. When $\bdim=3$ and $\poly=\trace 2$, the Cotton tensor $\Cotton{}{}{\polysub}$ agrees with the well-known Cotton-York tensor, \cite{Cotton1899,York71}.

The following theorem is proved in \cite{And84}.

\begin{theorem}\label{T:cohomology} Let $\Tstand{}{}$ be a natural, locally variational source form. Then $\Tstand{}{}$ can be expressed as
\begin{equation}
\Tstand{}{} =
\begin{cases} \ELop(\lambda),\qquad&\text{if $m\equiv0,1,2\mod 4$;}\\
\ELop(\lambda)+\Cotton{}{}\polysub(\metric),\qquad&\text{if $m\equiv3\mod 4$,}
\end{cases}
\end{equation}
where $\lambda = \Lagr\dvol$ is a natural Lagrangian form and $\Cotton{}{}\polysub$ is the generalized Cotton tensor associated with a uniquely determined $\poly\in\invpoly{\frac{m+1}2}{\ortha{p,q}}$.
\end{theorem}

\section{Symmetries and conservation laws}
\label{sec:symcl}

In this section we analyze the relationship between the diffeomorphisms symmetries \eqref{eq:TensorDensity} and the differential constraints \eqref{eq:CovarDiv} expressing divergence-freeness in metric field theories. As is well known, the
commutation formula \eqref{eq:EulerEquivariance}, together with Noether's
theorems, implies that the Euler-Lagrange expression $\Tsf=\ELop(\lambda)$ of a Lagrangian form $\lambda$ with symmetries \eqref{E:LScalarDensity} possesses symmetries \eqref{eq:TensorDensity} and is, in addition, constrained by the identities \eqref{eq:CovarDiv}. The following result, which bypasses the Lagrangian and is stated directly in terms of the system of differential equations, is a slight but non-vacuous extension of the above conclusions furnished by the classical Noether's theorems; see \cite{And78b}.
\smallskip
\begin{proposition}
\label{pro:TlocvarSymCL}
Suppose that the source form $\Tsf=\Tstand ab(\gccoll k) d\gc ab\wedge\dvol$ is locally variational.
Then $\Tsf$ is natural, $\lieder{\pr\liftedvf{\vfc{}}}\Tsf=0$ for all $\liftedvf{\vfc{}}\in\difflift$, if and only if it is divergence-free, $D_b \Tstand a b + \Christoffel a b c \Tstand b c=0$.
\end{proposition}
\smallskip
\begin{proof}
  By assumption the \He operator $\Helm$ of $\Tsf$ vanishes, and so equation
  \eqref{eq:Liederformula} reduces to
  \begin{equation}\label{eq:RedLiederformula}
    \cL_{\pr\liftedvf{\vfc{}}}\Tsf =
    \ELop(\liftedvf{\vfc{},\text{ev}}\ianhook \Tsf)=-\ELh hk{2\vfc{c}_{,a}\gcl bc\Tstand ab+\vfc{c}\gcjl abc\Tstand ab}d\gc hk\wedge\dvol.
\end{equation}
Recall that the Euler-Lagrange operator annihilates total derivative expressions, $\ELop(D_aF)=0$ for all differential functions $F=F(\bcoord i,\gccoll k)$ and $a=1,\dots,\bdim$. This permits us to integrate the right-hand side of \eqref{eq:RedLiederformula} by parts to conclude that
\begin{equation*}
\begin{split}
\ELh hk{2\vfc{c}_{,a}&\gcl bc\Tstand ab+\vfc c\gcjl abc\Tstand ab}d\gc hk\wedge\dvol
\\
&=
-\ELh hk{\vfc c(2\gc bc D_a\Tstand ab+2\gcj bc a\Tstand ab-\vfc c\gcj abc\Tstand ab)}d\gc hk\wedge\dvol\\&=-2\ELh hk{\vfc c\gc bc(D_a\Tstand ab+\Christoffel b cd\Tstand cd)}d\gc hk\wedge\dvol.
\end{split}
\end{equation*}
Hence
\begin{equation}\label{E:NaturalVSDivFree}
\cL_{\pr\liftedvf{\vfc{}}}\Tsf =2\ELh hk{\vfc c\gc bc(D_a\Tstand ab+\Christoffel b cd\Tstand cd)}d\gc hk\wedge\dvol,
\end{equation}
which immediately establishes the Proposition in one direction.

It then remains to
prove that the condition $\cL_{\pr\liftedvf{\vfc{}}}\Tsf=0$ for all $\vfc{}\in\vfs{\R\bdim}$
implies that $D_a\Tstand ab+\Christoffel b cd\Tstand cd=0$.
For this, suppose that for some index $b_o$,
$D_a\Tstand a{b_o}+\Christoffel {b_o} cd\Tstand cd=F(\bcoord i,\gccoll l)$ is of order $l$ and that for some
$h$, $k$ and $J$ with $|J|=l$, \[(\wpartial hk J F)(\bcoordo i,\gccollo l)\neq 0, \qquad(\bcoordo i,\gccollo l)\in\J l \Gbundle.\]
Now choose
$\vfc c_o\in\vfs{\R\bdim}$ such that
\[
\dfrac{\partial^{|J|}\vfc c_o}{\partial
x^J}(x^i_o)\gco bc=\delta_{bb_o}, \qquad
\dfrac{\partial^{|K|}\vfc c_o}{\partial x^K}(x^i_o)\gco bc=0,
\quad\mbox{$K\neq J$}.\] Let $(\bcoordo i,\gccollo{2l})\in\J{2l}\Gbundle$ be any point projecting to $\gccollo l$. Then
\[
\begin{split}
&\EL hk{\vfc c_o\gc bc(D_a\Tstand ab+\Christoffel b cd\Tstand cd)}(\bcoordo i,\gccollo{2l})\\
    &\quad=\sum_{|I|\leq l}(-1)^{|I|}D_{I}\left(\wpartial hk I (\vfc c_o\gc bc(D_a\Tstand ab+\Christoffel b cd\Tstand cd))\right)(\bcoordo i,\gccollo{2l})
    %  \partial^{b,I}_\beta F \right)(x^i_o,\accoll{2l}_o)
= (\wpartial hk JF)(\bcoordo i,\gccollo{l})\neq 0,
  \end{split}
  \]
  which is a contradiction. Thus
  $D_a\Tstand a{b}+\Christoffel {b} cd\Tstand cd=h^b(\bcoord i)$ are functions of $x^i$ only. But due to translational invariance of the source form $\Tsf$, each $h^b$ must be constant.  Finally, by the definition of 
  %the covariant divergence \eqref{eq:CovarDiv} and 
  the total derivative operators \eqref{eq:totalder} and the Chrisoffel symbols $\Christoffel b cd$, the expressions $D_a\Tstand a{b}+\Christoffel {b} cd\Tstand cd$ vanish when evaluated on the jet of the constant metric $\metric = \diag(1,\dots,1,-1,\dots,-1)$, showing that $h^b=0$.
\end{proof}

Recall that a source form $\Tsf=\T abd\gc ab\wedge\dvol$ admits translational conservation laws if 
%there is a constant vector $\tvf = k^c\pd{}{\bcoord c}\neq0$ so that
\[
\ELop(\gcj ab p\Tstand ab)=0,\qquad\text{ for all $1\leq p\leq\bdim$.}
\]

\smallskip
\begin{proposition}\label{prop:symS1S2consC1}
  Suppose that a source form $\Tsf$ is natural and admits translational conservation laws. Then the covariant divergence of $\Tsf$ vanishes, $D_b\Tstand ab+\Christoffel a bc\Tstand bc=0$.
\end{proposition}
\smallskip
\begin{proof} We will prove this Proposition by showing that 
  \begin{equation}\label{eq:HelmVanish}
    \Helm(\liftedvf{\vfc{},}{}_{\text{ev}})=0\qquad\text{for all $\vfc{}\in\vfs{\R\bdim}$.}
    \end{equation}

  First, with $\vf = \liftedvf{\vfc{}}=\vfc a\pd{}{\bcoord a} -2\vfc c_{,a}\gcl bc
      \wpartialo ab{}\in\difflift$, equation \eqref{eq:HelmCompInv} becomes
\begin{equation}\label{eq:GaugeHelmCond}
{\pr\liftedvf{\vfc{}}}\Helmgcstand
+\sum_{|J|\geq|I|}(\wpartial cd I \liftedvf{\vfc{},ef,J})\Helmcomp ab ef J
-2\vfc {(a}_{,e}\Helmcomp {b)}e cd I
+{\vfc e_{,e}}\Helmacstand=0,
\end{equation}
where $\liftedvf{\vfc{},ef,J}=\pr\liftedvf{\vfc{}}(\gcj ef J)$.
The source form $\Tsf$ is also invariant under $\tvf_p=\partial/\partial\bcoord p$, so the hypothesis and 
  equation \eqref{eq:Liederformula} imply
  \begin{equation}
    \label{eq:TransHelmCond}
    \sum_{|I|\geq0}\gcj cd {Ip}\Helmgcstand=0,\qquad\text{for all $1\leq p\leq\bdim$}.
  \end{equation}
  Now apply the vector field $\pr\liftedvf{\vfc{}}$ to the above equation to see
  that
  \begin{equation}\label{eq:GaugeInvTransHelm}
    \sum_{|I|\geq0}  \liftedvf{\vfc{},cd,Ip} \Helmgcstand
    +\sum_{|I|\geq0}\gcj cd{Ip}\pr\liftedvf{\vfc{}}
    (\Helmacstand) = 0.
  \end{equation}
  Combining~\eqref{eq:GaugeHelmCond} and~\eqref{eq:GaugeInvTransHelm} we obtain
  \[
\begin{split}
  \sum_{|I|\geq0} \liftedvf{\vfc{},cd,Ip} \Helmgcstand &
  -\sum_{|I|\geq0}\gcj cd{Ip} \left(\sum_{|J|\geq|I|}(\wpartial cd I X_{\vfc{},ef,J})\Helmcomp ab ef J
-2\vfc {(a}_{,e}\Helmcomp {b)}e cd I
+{\vfc e_{,e}}\Helmacstand\right)\\
& =\sum_{|I|\geq0}  \liftedvf{\vfc{},cd,Ip} \Helmgcstand -
  \sum_{|I|\geq0}\sum_{|J|\geq|I|}\gcj cd{Ip} (\wpartial cd I X_{\vfc{},ef,J})\Helmcomp ab ef J=0,
  \end{split}
  \]
where we used \eqref{eq:TransHelmCond} twice.
On account of the definition of the total derivative operators
  \eqref{eq:totalder}, the above equation can be written as
  \begin{equation}\label{eq:SimplifyI}
  \sum_{|I|\geq0} \liftedvf{\vfc{},cd,Ip} \Helmgcstand 
  -\sum_{|I|\geq0}D_p \liftedvf{\vfc{},cd,I}\Helmcomp ab cd I
  +\sum_{|I|\geq0}\pd{\liftedvf{\vfc{},cd,I}}{\bcoord p\hfill}\Helmcomp ab cd I=0.
\end{equation}
By the standard prolongation formula \eqref{eq:prolongation},
\[
X_{\vfc{},cd,Ip} = D_p X_{\vfc{},cd,I}-\pd{\vfc q}{\bcoord p}\gcj cd{Iq}.
\]
Thus equation \eqref{eq:SimplifyI} simplifies to
\[
\begin{split}
-\sum_{|I|\geq0}&  \pd{\vfc q}{\bcoord p}\gcj cd{Iq} \Helmgcstand+
\sum_{|I|\geq0}\pd{X_{\vfc{},cd,I}}{\bcoord p\hfill}\Helmcomp ab cd I\\&=\sum_{|I|\geq0}\pd{X_{\vfc{},cd,I}}{\bcoord p\hfill}\Helmcomp ab cd I
=\sum_{|I|\geq0}{X_{\pd{\vfc{}}{\bcoord p\hfill},cd,I}}\Helmcomp ab cd I=0,
\end{split}
\]
where we again used \eqref{eq:TransHelmCond}.
The vector field $\vfc{}\in\vfs{\R\bdim}$ is arbitrary, which allows us to conclude that
  \[
  \sum_{|I|\geq0}{\liftedvf{{\vfc{}},cd,I}}\Helmcomp ab cd I=0,\qquad\text{for all $\vfc{}\in\vfs{\R\bdim}$.}
  \]
But on account \eqref{eq:TransHelmCond},
\[
\sum_{|I|\geq0}{X_{{\vfc{}},cd,I}}\Helmcomp ab cd I=\sum_{|I|\geq0}(D_I{X_{{\vfc{}},\text{ev},cd}}+\vfc q\gcj cd{Iq})\Helmcomp ab cd I=\sum_{|I|\geq0}{D_IX_{{\vfc{}},\text{ev},cd}}\Helmcomp ab cd I=0,
\]
that is, equation \eqref{eq:HelmVanish} holds.

  Due to the $\difflift$-invariance of the source form $\Tsf$ and condition
  \eqref{eq:HelmVanish}, the Lie derivative formula \eqref{eq:Liederformula} now yields
  \[\ELop(\vfc c\gc bc(D_a\Tstand ab+\Christoffel b cd\Tstand cd))=0,\qquad \text{for all $\vfc{c}\in\vfs{\R\bdim}$.}
\]
  Next we continue as in the second part of the proof of
  Proposition~\ref{pro:TlocvarSymCL} to conclude that the source form $\Tsf$ is divergence-free.
\end{proof}

\section{Proof of Theorem~1}
\label{sec:proof}

The proof of Theorem~\ref{th:main1} relies on the following Lemma, which is a
special case of a more general result presented in \cite{AnDu85,Olv83}.
\smallskip
\begin{lemma}\label{L:polT}
  Let $\Tsf=\Tstand a b d\gcstand\wedge\dvol$, $\Tstand a b=\Tstand {(a}{b)}(x^i,\gccoll k)$, be a
  $k$-th order source form, where $k\geq1$, and assume that the covariant divergence $D_b\Tstand a b+\Christoffel a bc\Tstand bc=0$ vanishes identically. Then the component functions $\Tstand a b$ are
  polynomials in the $k$-th order derivative variables $\gcj cd {i_1\cdots
    i_k}$ of degree at most $\bdim-1$, where $\bdim$ is the number of the independent variables.
\end{lemma}
\smallskip
\begin{proof}
  By assumption,
  \begin{equation}\label{eq:DivergenceCond}
    \pd{\Tstand a b}{x^b} + \sum_{|I|\geq0}\gcj cd {Ib} \wpartial cd I
    \Tstand a b+
    \Christoffel a cd\T cd = 0.
  \end{equation}
  Now terms in \eqref{eq:DivergenceCond} involving the order $k+1$ variables
  $\gcj  ab J$, $|J|=k+1$, yield the equations
  \begin{equation}\label{eq:polT-1}
  \wpartial cd {(I}\Tstand {b)} a=0.
  \end{equation}
  Write $\copartial b \beta k X$ for the partial differential operator
  \begin{equation}\label{eq:1}
    \copartial ab k X = \sum_{|I|=k} X_I\wpartial a b I
   = \sum\wpartial ab{i_1\cdots i_k}X_{i_1}\cdots X_{i_k},
   %_{1\leq i_1,\dots,i_k\leq \bdim} 
  \end{equation}
  where $X = (X_1,\dots,X_\bdim)\in(\R \bdim)^*$ is a covector on $\R \bdim$. Then
  equation~\eqref{eq:polT-1} is equivalent to
  \begin{equation}
  \label{eq:divcond}
  X_b \copartial cd k X \Tstand a b = 0\qquad\text{for all $X$.}
  \end{equation}

  Next let $G^{a,c_1d_1\cdots c_\bdim d_\bdim}$ denote the
  mappings
  \begin{displaymath}
    G^{a,c_1d_1\cdots c_\bdim d_\bdim}(X^1,\ldots,X^\bdim,Y)=
    \copartial {c_1} {d_1} k {X^1} \cdots
    \copartial {c_\bdim} {d_\bdim} k {X^\bdim} \Tstand a b Y_b.
  \end{displaymath}
  The operator $\Tstand a b$ is polynomial in $\gcj ab I$, $\abs{I}=k$, of
  degree at most $\bdim-1$ if and only if all the mappings
  $G^{a,c_1d_1\cdots c_\bdim d_\bdim}$ vanish identically. But
  by~\eqref{eq:divcond} and linearity in $Y$, the equation
  \begin{displaymath}
    G^{a,c_1d_1\cdots c_\bdim d_\bdim}(X^1,\ldots,X^\bdim,Y) = 0
  \end{displaymath}
  holds whenever $Y$ is a linear combination of the covectors $X^1$, \dots,
  $X^\bdim$. Consequently $G^{a,c_1d_1\cdots c_\bdim d_\bdim}$
  vanishes for almost all $X^1$, \dots, $X^\bdim$, $Y\in(\R n)^*$. By continuity,
  $G^{a,c_1d_1\cdots c_\bdim d_\bdim}$ must vanish
  identically.
\end{proof}

We next employ the Lie derivative formula \eqref{eq:Liederformula} to derive key identities among the components of the \He operator associated with a third order natural, divergence free source form. On account of these identities and the integrability conditions of Proposition \ref{P:HelmIntCond}, the proof of the first part of  Theorem \ref{th:main1} reduces to showing only that the third order components of the \He operator vanish.

\begin{proposition}\label{P:HelmConsCondO} Suppose that $\Tsf=\T ab d\gc ab\wedge\dvol$ is a third order natural, divergence free source form, that is,
\[
\text{$\lieder{\pr\liftedvf{\vfc{}}}\Tsf=0$,\hspace{5pt} for all $\liftedvf{\vfc{}}\in\difflift$,\quad and\quad $D_b\T ab+\Christoffel abc\T bc=0$.}
\]
Then the components $\Helmcompo abcd$, $\Helmcomp abcdi$, $\Helmcomp abcd{ij}$,  $\Helmcomp abcd{ijk}$ of the \He operator $\Helm$ associated with $\Tsf$ satisfy the following identities.
\openup4pt
\begin{subequations}\label{eq:HelmConsCondO}
\begin{align}
\gcj cde\Helmcompo abcd+\gcj cd{ie}\Helmcomp abcdi+\gcj cd{ije}\Helmcomp abcd{ij}+\gcj cd{ijke}\Helmcomp abcd{ijk}&=0,\label{eq:HelmConsCondO0}\\
2\gc ce\Helmcompo abcf+2\gcj cei\Helmcomp abcfi+\gcj cde\Helmcomp abcdf+2\gcj ce{ij}\Helmcomp abcf{ij}
\qquad\qquad\qquad&\nonumber\\
+2\gcj cd{ie}\Helmcomp abcd{if}+2\gcj ce{ijk}\Helmcomp abcf{ijk}+3\gcj cd{ije}\Helmcomp abcd{ijf}&=0,\label{eq:HelmConsCondO1}\\
2\gc ce\Helmcomp abc{(d}{i)}+4\gcj cej\Helmcomp abc{(d}{i)j}+\gcj cfe\Helmcomp abcf{di}\qquad\qquad\qquad\nonumber\\
+6\gcj ce{jk}\Helmcomp abc{(d}{i)jk}+3\gcj cf{je}\Helmcomp abcf{dij}&=0,\label{eq:HelmConsCondO2}\\
2\gc ce\Helmcomp abc{(d}{ij)}+6\gcj cek\Helmcomp abc{(d}{ij)k}+\gcj cke\Helmcomp abck{ijd}&=0,\label{eq:HelmConsCondO3}\\
\Helmcomp abc{(d}{ijk)}&=0.\label{eq:HelmConsCondO4}
\end{align}
\end{subequations}
\openup-4pt
\end{proposition}

\begin{proof} By the Lie derivative formula \eqref{eq:Liederformula}, 
\begin{equation}\label{eq:DiffHelmConds}
\Helm(\liftedvf{\vfc{},}{}_{\text{ev}})=(D_I\liftedvf{\vfc{},}{}_{\text{ev,}cd})\Helmcomp abcdI d\gc ab\wedge\dvol=0,
\end{equation}
for all $\liftedvf{\vfc{}}\in\difflift$, where
\[
\liftedvf{\vfc{},}{}_{\text{ev}} = \liftedvfevc{\vfc{}}cd\wpartialo cd =-(2\vfc e_{,(c}\gcl{d)}e+\vfc e\gcjl cde)\wpartialo cd.
\]
Let $\mu^{cd,I} = \mu^{(cd),(I)}$, $|I|\geq0$, be multivectors. We compute
\begin{equation}\label{eq:prolongation1}
\begin{split}
D_I(\vfc e_{,c}\gcl de)\mu^{cd,I} 
&= \sum_{|I|,|J|\geq0}\binom{|I|+|J|}{|I|}\vfc e_{,cI}\gcjl deJ\mu^{cd,IJ}\\
&=\sum_{|I|>0}\vfc e_{,I}\sum_{|J|\geq0}\binom{|I|+|J|-1}{|I|-1}\gcjl deJ\mu^{d(i_1,i_2\cdots i_p)J},
\end{split}
\end{equation}
and
\begin{equation}\label{eq:prolongation2}
D_I(\vfc e\gcjl cde)\mu^{cd,I}=\sum_{|I|\geq0}\vfc e_{,I}\sum_{|J|\geq0}\binom{|I|+|J|}{|I|}
\gcjl cd{Je}\mu^{cd,IJ}.
\end{equation}

The derivatives $\vfc e_{,I}$ are independent, so in light of \eqref{eq:prolongation1}, \eqref{eq:prolongation2}, the coefficient of $\vfc e$ in \eqref{eq:DiffHelmConds} yields the equation 
\begin{equation}\label{eq:CharDers0}
\sum_{|J|\geq0}\gcjl cd{Je}\Helmcomp abcdJ=0,
\end{equation} 
while for $|I|>0$, we obtain
\begin{equation}\label{eq:CharDers1}
\sum_{|J|\geq0}\bigg[\binom{|I|+|J|-1}{|I|-1}2\gcjl deJ\Helmcomp abd{(i_1}{i_2\cdots i_p)J}+\binom{|I|+|J|}{|I|}\gcjl cd{Je}\Helmcomp abcd{IJ}\bigg]=0.
\end{equation}
Keeping in mind that $\Helmcomp abcdI=0$ for $|I|\geq4$, equations \eqref{eq:HelmConsCondO} follow from \eqref{eq:CharDers0} and \eqref{eq:CharDers1} by inspection. 
\end{proof}

We note that if $\mu^{abdc}=\mu^{(ab)(cd)}$ is a valence 4 tensor satisfying the cyclic identity $\mu^{a(bcd)}=0$, then 
\begin{equation}\label{eq:SymmPair}
\mu^{abcd}-\mu^{cdab}= \dfrac32(\mu^{a(bcd)}+\mu^{b(acd)}-\mu^{c(abd)}-\mu^{d(abc)})=0,
\end{equation}
so that $\mu^{abcd}$ is symmetric under the interchange of the pairs of indices $ab$ and $cd$.

\begin{corollary}\label{C:ThirdOrderVanish} Suppose that $\Tsf=\T ab d\gc ab\wedge\dvol$ is a third order natural, divergence free source form. Then $\Tsf$ is locally variational provided that the third order components $\Helmcomp abcd{ijk}$ of the \He operator $\Helm$ associated with $\Tsf$ vanish.
\end{corollary}

\begin{proof} We shall show that as a consequence of the assumptions, all the components of the \He operator $\Helm$ vanish identically. The proof is based on identities \eqref{eq:HelmIntCond} and \eqref{eq:HelmConsCondO} of Propositions \ref{P:HelmIntCond} and \ref{P:HelmConsCondO}, respectively.

We first observe that by \eqref{eq:HelmIntCond2}, the second order components $\Helmcomp abcd{ij}$ are skew-symmetric under the interchange of the index pairs $ab$ and $cd$. Due to the assumptions, equation \eqref{eq:HelmConsCondO3} reduces to $\Helmcomp abc{(d}{ij)}=0$. Consequently, identity \eqref{eq:SymmPair} applied to the index pairs $cd$ and $ij$ shows that the $\Helmcomp abcd{ij}$ are also symmetric under the interchange of these pairs. Thus the second order components $\Helmcomp abcd{ij}$ must vanish identically.

By the above, equation \eqref{eq:HelmConsCondO2} now implies that $\Helmcomp abc{(d}{i)}=0$. Thus $\Helmcomp abcdi$ is symmetric and skew-symmetric in overlapping pairs of indices and hence must vanish. Finally, equation \eqref{eq:HelmConsCondO1} presently shows that $\Helmcompo abcd=0$ for all $a$, $b$, $c$, $d$.
\end{proof}

 We will use a semi-colon to indicate differentiation with respect to the weighted partial derivative operators \eqref{eq:wpd}, so that, for example, 
\[
\wpartial ab{ij}\wpartial cd{klm}F=F^{;ab,ij;cd,klm},
\]
for a differential function $F$  on $\J\infty\Gbundle$.

We will use the following result in the course of the proof of Theorem \ref{th:main1}.
\begin{proposition}\label{p:HelmCompExpression}
 Let $\Tsf=\T abd\gc ab\wedge\dvol$ be a third order divergence-free source form and let $\Helmcomp abcd{i_1i_2}$ denote the second order components of the \He operator associated with $\Tsf$. Then
\begin{align}
(i)&\qquad\Tm a{b}{{c_1}{c_2,}{i_1i_2(i_3|};d_1d_2,|j_1j_2j_3)}=\Tm {i_1}{i_2}{{c_1}{c_2,}{ab(i_3|};d_1d_2,|j_1j_2j_3)};\label{eq:TIdentity1}\\
(ii)&\qquad\Tm {(a}{b|}{c_1c_2,|i_1)i_2i_3} = \frac13\Tm {i_2}{i_3}{c_1c_2,abi_1};\hskip1.5in\label{eq:TIdentity2}\\
%(iii)&\qquad\Helmcompm abcd{i_1i_2}{j_1j_2,k_1k_2k_3k_4}=3\Tm{i_1}{i_2}{ab,cd(k_1|;j_1j_2,|k_2k_3k_4)}
%\label{eq:TIdentity3}.\\
(iii)&\qquad\Helmcompm ab{c_1}{c_2}{i_1i_2}{d_1d_2,j_1j_2j_3j_4}=3\Tm{i_1}{i_2}{ab,{c_1}{c_2}(j_1|;d_1d_2,|j_2j_3j_4)}\label{eq:TIdentity3},
\end{align}
where indices enclosed within vertical bars are omitted in the symmetrization. 
%\begin{equation}
%\Tm a{(b}{|c_1c_2,|i_1i_2)(i_3|;d_1d_2,|j_1j_2j_3)} =0.
%\end{equation}
\end{proposition}
\begin{proof} Since $\Tsf$ is divergence free and of order 3, equation \eqref{eq:DivergenceCond} yields
\begin{equation}\label{eq:DivFreeThirdOrder}
%\wpartial{c_1}{c_2}{i_1i_2i_3i_4}(D_b\T ab+\Christoffel a bc\T bc)=
%\delta_b^{(i_1|}\wpartial{c_1}{c_2}{|i_2i_3i_4)}\T ab=
\Tm a{(b|}{{c_1}{c_2}{|i_1i_2i_3)}}=0.
\end{equation}
Consequently,
\begin{equation}\label{eq:norsu}
\Tm a{(b|}{{c_1}{c_2}{|i_1i_2)i_3}}=-\dfrac13 \Tm a{i_3}{{c_1}{c_2},{bi_1i_2}},
\end{equation}
which, on account of \eqref{eq:DivFreeThirdOrder}, implies that
\begin{equation*}
\begin{split}
\Tm a{(b|}{{c_1}{c_2}{|i_1i_2)(i_3|};d_1d_2,|j_1j_2j_3)}&=-\dfrac13 \Tm a{(i_3|}{{c_1}{c_2},{bi_1i_2};d_1d_2,|j_1j_2j_3)}\\
&=-\dfrac13 \wpartial{c_1}{c_2}{bi_1i_2}\Tm a{(i_3|}{d_1d_2,|j_1j_2j_3)}=0.
\end{split}
\end{equation*}
Equation \eqref{eq:SymmPair} now yields \eqref{eq:TIdentity1}.

Next, due to \eqref{eq:DivFreeThirdOrder}, we have that
\begin{equation}
\Tm {(a}{b|}{c_1c_2,|i_1)i_2i_3}
=-\underset{\{abi_1\}}{\Sym}\;\Tm {(i_2|}{a}{c_1c_2,bi_1|i_3)}
=\dfrac13\Tm{i_2}{i_3}{c_1c_2,ab{i_1}},
\end{equation}
where in the second step we used equation \eqref{eq:norsu}.

Finally, in light of the commutation formula \eqref{eq:totaldercomm},
\begin{equation*}
\begin{split}
\Helmcompm ab{c_1}{c_2}{i_1i_2}{d_1d_2,j_1j_2j_3j_4}
&=\wpartial{d_1}{d_2}{j_1j_2j_3j_4}\big(\wpartial {c_1}{c_2}{i_1i_2}\T ab-\wpartial ab{i_1i_2}\T {c_1}{c_2}+3D_{i_3}(\wpartial ab{i_1i_2i_3}\T {c_1}{c_2})\big)\\
&= 3{\Tm {c_1}{c_2}{ab,i_1i_2(j_1|;d_1d_2|j_2j_3j_4)}}=3{\Tm {i_1}{i_2}{ab,{c_1}{c_2}(j_1|;d_1d_2|j_2j_3j_4)}},
\end{split}
\end{equation*}
which completes the proof of the Proposition.
\end{proof}
 
\begin{proof}[Proof of Theorem~\protect\ref{th:main1}]

Expression \eqref{eq:InverseProblem} is a straightforward consequence of Theorem \ref{T:cohomology} for locally variational source forms. Thus it suffices, by  Corollary~\ref{C:ThirdOrderVanish}, to show that the third order components $\Helmcomp abcd{ijk}$ of the \He operator $\Helm$ vanish. 

In order to streamline our notation, we let $\pmi=(p_1,p_2)$ and $\qmi=(q_1,q_2,q_3)$ denote multi-indices of integers $1\leq p_t\leq m$, $1\leq q_t\leq m$ of respective lengths 2 and 3, and write    
$\wpartialm\pmi\qmi=\wpartial{p_1}{p_2}{q_1q_2q_3}$ in what follows. 
We also employ the notation
\[
F^{;\pmi_1\qmi_1;\pmi_2\qmi_2;\cdots;\pmi_s\qmi_s}=\wpartialm{\pmi_1}{\qmi_1}\wpartialm{\pmi_2}{\qmi_2}\cdots\wpartialm{\pmi_s}{\qmi_s}F
\]
for the repeated derivatives of $F$ with respect to the weighted partial derivative operators $\wpartialm{\pmi_u}{\qmi_u}=\wpartial{p_{u1}}{p_{u2}}{q_{u1}q_{u2}q_{u3}}$.

By Lemma~\ref{L:polT}, the $\bdim$-fold derivatives $\Tstand a{b;\pmi_1\qmi_1;\pmi_2\qmi_2;\cdots;\pmi_m\qmi_m}=0$, and consequently, the third order components $\Helmcomp abcd{ijk}=\wpartial cd{ijk}\T ab+\wpartial ab{ijk}\T cd$ of the Helmholtz operator satisfy
\[
\Helmcompm abcd{ijk}{\pmi_1\qmi_1;\pmi_2\qmi_2;\cdots;\pmi_{m-1}\qmi_{m-1}}=0.
\]
We will prove by induction that $\Helmcompm abcd{ijk}{\pmi_1\qmi_1;\pmi_2\qmi_2;\cdots;\pmi_{r}\qmi_{r}}=0$ for all $0\leq r\leq\bdim-1$.

In order to carry out the induction step we will assume that 
\begin{equation}\label{eq:IndHypo}
\Helmcompm abcd{ijk}{\pmi_1\qmi_1;\pmi_2\qmi_2;\cdots;\pmi_{r+1}\qmi_{r+1}}=0,
\quad\text{ for some $0\leq r\leq m-2$.}
\end{equation}
Our first goal is to show that 
\begin{subequations}\label{eq:IndStep1}
\begin{align}
\Helmcompom abcd{l_1l_2,n_1n_2n_3 n_4;\pmi_1\qmi_1;\pmi_2\qmi_2;\cdots;\pmi_{r}\qmi_{r}}&=0,\label{eq:IndStep10}\\
\Helmcompm abcd{i}{l_1l_2,n_1n_2n_3 n_4;\pmi_1\qmi_1;\pmi_2\qmi_2;\cdots;\pmi_{r}\qmi_{r}}&=0,\label{eq:IndStep11}\\
\Helmcompm abcd{ij}{l_1l_2,n_1n_2n_3 n_4;\pmi_1\qmi_1;\pmi_2\qmi_2;\cdots;\pmi_{r}\qmi_{r}}&=0.\label{eq:IndStep12}
\end{align}
\end{subequations}
For this, we start by applying the operator $\wpartial{l_1}{l_2}{n_1n_2n_3 n_4}$ to \eqref{eq:HelmIntCond2} to see that
\[
\Helmcompm abcd{ij}{l_1l_2,n_1n_2n_3n_4} + \Helmcompm cdab{ij}{l_1l_2,n_1n_2n_3n_4}=
3\Helmcompm abcd{ij(n_1|}{l_1l_2,|n_2n_3 n_4)}.
\]
Thus by the induction assumption \eqref{eq:IndHypo},
\[
\Helmcompm abcd{ij}{l_1l_2,n_1n_2n_3 n_4;\pmi_1\qmi_1;\pmi_2\qmi_2;\cdots;\pmi_{r}\qmi_{r}} + \Helmcompm cdab{ij}{l_1l_2,n_1n_2n_3 n_4;\pmi_1\qmi_1;\pmi_2\qmi_2;\cdots;\pmi_{r}\qmi_{r}}=0,
\]
that is, the components $\Helmcompm abcd{ij}{l_1l_2,n_1n_2n_3n_4}$ are skew-symmetric under the interchange of the index pairs $ab$ and $cd$.
On the other hand, repeated differentiation of \eqref{eq:HelmConsCondO3} yields the equation
\[
\Helmcompm abc{(d}{ij)}{l_1l_2,n_1n_2n_3 n_4;\pmi_1\qmi_1;\pmi_2\qmi_2;\cdots;\pmi_{r}\qmi_{r}}=0.
\]
Thus by \eqref{eq:SymmPair}, the components $\Helmcompm abcd{ij}{l_1l_2,n_1n_2n_3n_4}$ are symmetric in the pairs $cd$ and $ij$. It now follows that \eqref{eq:IndStep12} holds.

By differentiating \eqref{eq:HelmConsCondO2} we see that
\[
\Helmcompm abc{(d}{i)}{l_1l_2,n_1n_2n_3 n_4;\pmi_1\qmi_1;\pmi_2\qmi_2;\cdots;\pmi_{r}\qmi_{r}}=0,
\]
where we used \eqref{eq:IndStep12} and the identity $\Helmcompm abcd{ijk}{l_1l_2,n_1n_2n_3n_4}=0$. 
Consequently, \eqref{eq:IndStep11} holds. Finally, the first equation \eqref{eq:IndStep10} now follows from \eqref{eq:HelmConsCondO1}, again by differentiation.

Next apply the differential operators $\wpartial{l_1}{l_2}{n_1n_2n_3 n_4}$, $\wpartialm{\pmi_1}{\qmi_1}$,\dots,$\wpartialm{\pmi_r}{\qmi_r}$,
to equation \eqref{eq:HelmConsCondO0}. Keeping in mind that the components $\Helmcomp abcd{ijk}$ are of order 3 we conclude that
\begin{equation}
\begin{split}
\gcj cde&\Helmcompom abcd{l_1l_2,n_1n_2n_3 n_4;\pmi_1\qmi_1;\cdots;\pmi_r\qmi_r}+\gcj cd{ie}\Helmcompm abcdi{l_1l_2,n_1n_2n_3n_4;\pmi_1\qmi_1;\cdots;\pmi_r\qmi_r}\\
&+\gcj cd{ije}\Helmcompm abcd{ij}{l_1l_2,n_1n_2n_3 n_4;\pmi_1\qmi_1;\cdots;\pmi_r\qmi_r}\\
&+\sum_{t=1}^r\delta_e^{(q_{t1}|}\Helmcompm ab{p_{t1}}{p_{t2}}{|q_{t2}q_{t3})}{l_1l_2,n_1n_2n_3 n_4;\pmi_1\qmi_1;\cdots\widehat{\pmi_t\qmi_t}\cdots;\pmi_r\qmi_r}\\&+\delta_{e}^{(n_1|}\Helmcompm ab{l_1}{l_2}{|n_2n_3n_4)}{\pmi_1\qmi_1;\cdots;\pmi_r\qmi_r}=0,
\end{split}
\end{equation}
where the hat indicates omission. On account of \eqref{eq:IndStep1}, the above equation reduces to 
\begin{equation}\label{eq:IndStep13}
\begin{split}
\sum_{t=1}^r\delta_e^{(q_{t1}|}&\Helmcompm ab{p_{t1}}{p_{t2}}{|q_{t2}q_{t3})}{l_1l_2,n_1n_2n_3 n_4;\pmi_1\qmi_1;\cdots\widehat{\pmi_t\qmi_t}\cdots;\pmi_r\qmi_r}\\&\qquad\quad+\delta_{e}^{(n_1|}\Helmcompm ab{l_1}{l_2}{|n_2n_3n_4)}{\pmi_1\qmi_1;\cdots;\pmi_r\qmi_r}=0.
\end{split}
\end{equation}

Our next goal is to show that $\Helmcompm abcd{ef}{{l_1l_2,n_1n_2n_3 n_4;\pmi_1\qmi_1;\cdots;\pmi_{r-1}\qmi_{r-1}}}
$. As a result, equation \eqref{eq:IndStep13} becomes 
\begin{equation}\label{eq:IndStep121}
\delta_{e}^{(n_1|}\Helmcompm ab{l_1}{l_2}{|n_2n_3n_4)}{\pmi_1\qmi_1;\cdots;\pmi_r\qmi_r}=0,
\end{equation} 
which will immediately imply the induction step.

To proceed, we first contract in the index pair $e$, $q_{11}$ in  \eqref{eq:IndStep13}, which yields the equation
\begin{equation}\label{eq:IndStep14}
\begin{split}
\frac{\bdim+2}{3}&\Helmcompm ab{p_{11}}{p_{12}}{q_{12}q_{13}}{l_1l_2,n_1n_2n_3 n_4;\pmi_2\qmi_2;\cdots;\pmi_r\qmi_r}\\
&+\sum_{t=2}^r\Helmcompm ab{p_{t1}}{p_{t2}}{(q_{t1}q_{t2}|}{l_1l_2,n_1n_2n_3 n_4;\pmi_1,|q_{t3})q_{12}q_{13};\pmi_2\qmi_2;\cdots\widehat{\pmi_t\qmi_t}\cdots;\pmi_r\qmi_r}\\
&+\Helmcompm ab{l_1}{l_2}{(n_1n_2n_3|}{\pmi_1,|n_4)q_{12}q_{13};\pmi_2\qmi_2;\cdots;\pmi_r\qmi_r}=0.
\end{split}
\end{equation} 
On account of \eqref{eq:TIdentity3}, we have
\begin{equation*}
\Helmcompm ab{p_{11}}{p_{12}}{q_{12}q_{13}}{l_1l_2,n_1n_2n_3 n_4;\pmi_2\qmi_2;\cdots;\pmi_r\qmi_r}=3\Tm {q_{12}}{q_{13}}{ab,{p_{11}}{p_{12}(n_1|};{l_1l_2,|n_2n_3n_4);\pmi_2\qmi_2;\cdots;\pmi_r\qmi_r}}.
\end{equation*}
We next use \eqref{eq:TIdentity2} and \eqref{eq:TIdentity3} to compute 
\begin{align*}
&\Helmcompm ab{p_{t1}}{p_{t2}}{(q_{t1}q_{t2}|}{l_1l_2,n_1n_2n_3 n_4;\pmi_1,|q_{t3})q_{12}q_{13};\pmi_2\qmi_2;\cdots\widehat{\pmi_t\qmi_t}\cdots;\pmi_r\qmi_r}
\\
&\qquad =3\underset{\{n_1\cdots n_4\}}\Sym\Tm {(q_{t1}}{q_{t2}|}{ab,{p_{t1}}{p_{t2}n_1};{l_1l_2,n_2n_3n_4;\pmi_1,|q_{t3})q_{12}q_{13};\pmi_2\qmi_2;\cdots\widehat{\pmi_t\qmi_t}\cdots;\pmi_r\qmi_r}}\\
&\qquad =3\Tm {(q_{t1}}{q_{t2}|}{\pmi_1,|q_{t3})q_{12}q_{13};ab,p_{t1}{p_{t2}(n_1|};{l_1l_2,|n_2n_3n_4);\pmi_2\qmi_2;\cdots\widehat{\pmi_t\qmi_t}\cdots;\pmi_r\qmi_r}}\\
&\qquad=\Tm {q_{12}}{q_{13}}{\pmi_1,\qmi_t;ab,{p_{t1}}{p_{t2}(n_1|};{l_1l_2,|n_2n_3n_4);\pmi_2\qmi_2;\cdots\widehat{\pmi_t\qmi_t}\cdots;\pmi_r\qmi_r}}.
\end{align*}
Similarly, by \eqref{eq:ThirdOrderHe} and \eqref{eq:TIdentity1},
\begin{align*}
&\Helmcompm ab{l_1}{l_2}{(n_1n_2n_3|}{\pmi_1,|n_4)q_{12}q_{13};\pmi_2\qmi_2;\cdots;\pmi_r\qmi_r}\\
&\qquad=\Tm ab{{l_1}{l_2},{(n_1n_2n_3|};\pmi_1,|n_4)q_{12}q_{13};\pmi_2\qmi_2;\cdots;\pmi_r\qmi_r}+
\Tm{l_1}{l_2}{ab,{(n_1n_2n_3|};\pmi_1,|n_4)q_{12}q_{13};\pmi_2\qmi_2;\cdots;\pmi_r\qmi_r}\\
&\qquad=\Tm {q_{12}}{q_{13}}{{l_1}{l_2},{(n_1n_2n_3|};\pmi_1,|n_4)ab;\pmi_2\qmi_2;\cdots;\pmi_r\qmi_r}+
\Tm{q_{12}}{q_{13}}{ab,{(n_1n_2n_3|};\pmi_1,|n_4){l_1}{l_2};\pmi_2\qmi_2;\cdots;\pmi_r\qmi_r}.
\end{align*}
By the above, equation \eqref{eq:IndStep14} becomes
\begin{equation}\label{eq:IndStep15}
\begin{split}
(\bdim+2)&\Tm {q_{12}}{q_{13}}{ab,{p_{11}}{p_{12}(n_1|};{l_1l_2,|n_2n_3n_4);\pmi_2\qmi_2;\cdots;\pmi_r\qmi_r}}\\
&+
\sum_{t=2}^r\Tm {q_{12}}{q_{13}}{ab,{p_{t1}}{p_{t2}(n_1|};{l_1l_2,|n_2n_3n_4);\pmi_1\qmi_t;\pmi_2\qmi_2;\cdots\widehat{\pmi_t\qmi_t}\cdots;\pmi_r\qmi_r}}\\
&+\Tm {q_{12}}{q_{13}}{{l_1}{l_2},{(n_1n_2n_3|};\pmi_1,|n_4)ab;\pmi_2\qmi_2;\cdots;\pmi_r\qmi_r}\\
&+
\Tm{q_{12}}{q_{13}}{ab,{(n_1n_2n_3|};\pmi_1,|n_4){l_1}{l_2};\pmi_2\qmi_2;\cdots;\pmi_r\qmi_r}=0.
\end{split}
\end{equation}
We have thus reduced \eqref{eq:IndStep14} into an equation for the derivatives of a single component $\T{q_{12}}{q_{13}}$ of the source form $\Tsf$.

Now assume that equations \eqref{eq:IndStep15} admitted a non-trivial solution in the derivatives of $\T{q_{12}}{q_{13}}$. Choose a term $\Tm {q_{12}}{q_{13}}{ab,{p_{11}}{p_{12}(n_1|};{l_1l_2,|n_2n_3n_4);\pmi_2\qmi_2;\cdots;\pmi_r\qmi_r}}$ with maximal absolute value amongst the symmetrized derivatives. Keeping in mind that $r\leq m-2$ by the induction assumption, equation \eqref{eq:IndStep15} implies that
\begin{equation*}
\begin{split}
&|\Tm {q_{12}}{q_{13}}{ab,{p_{11}}{p_{12}(n_1|};{l_1l_2,|n_2n_3n_4);\pmi_2\qmi_2;\cdots;\pmi_r\qmi_r}}|\\
&\leq \frac1{\bdim+2}\bigg(\sum_{t=2}^r|\Tm {q_{12}}{q_{13}}{ab,{p_{t1}}{p_{t2}(n_1|};{l_1l_2,|n_2n_3n_4);\pmi_1\qmi_t;\pmi_2\qmi_2;\cdots\widehat{\pmi_t\qmi_t}\cdots;\pmi_r\qmi_r}}|\\
&+
|\Tm {q_{12}}{q_{13}}{{l_1}{l_2},{(n_1n_2n_3|};\pmi_1,|n_4)ab;\pmi_2\qmi_2;\cdots;\pmi_r\qmi_r}|+
|\Tm{q_{12}}{q_{13}}{ab,{(n_1n_2n_3|};\pmi_1,|n_4){l_1}{l_2};\pmi_2\qmi_2;\cdots;\pmi_r\qmi_r}|\bigg)\\
&\leq \frac{\bdim-1}{\bdim+2}|\Tm {q^2_1}{q^3_1}{ab,{p^1_1}{p^2_1(n_1|};{l_1l_2,|n_2n_3n_4);\pmi_2\qmi_2;\cdots;\pmi_r\qmi_r}}|.
\end{split}
\end{equation*}  
Thus necessarily
\begin{equation}
\Helmcompm abcd{ef}{{l_1l_2,n_1n_2n_3 n_4;\pmi_1\qmi_1;\cdots;\pmi_{r-1}\qmi_{r-1}}}=\Tm {e}{f}{ab,{c}{d(n_1|};{l_1l_2,|n_2n_3n_4);\pmi_1\qmi_1;\cdots;\pmi_{r-1}\qmi_{r-1}}}=0,
\end{equation}
for all $a,b,c,d,e,f,l_1,l_2,n_1,\dots,n_4$, $\pmi_1,\qmi_1,\dots, \pmi_{r-1},\qmi_{r-1}$. Hence \eqref{eq:IndStep13} reduces to
\[
\delta_{e}^{(n_1|}\Helmcompm ab{l_1}{l_2}{|n_2n_3n_4)}{\pmi_1\qmi_1;\cdots;\pmi_r\qmi_r}=0,
\]
which immediately implies that 
$\Helmcompm ab{l_1}{l_2}{n_1n_2n_3}{\pmi_1\qmi_1;\cdots;\pmi_r\qmi_r}=0$, completing the induction step.

 In conclusion, the third order components $\Helmcomp abcd{jkl}$ of the \He operator associated with $\Tsf$ vanish identically. By Corollary \ref{C:ThirdOrderVanish} the \He operator $\Helm$ must also vanish identically, and hence $\Tsf$ is locally variational.  
\end{proof}

\section{Discussion}\label{sec:discussion}

In this paper we prove that a system of third order natural, divergence-free differential equations for the components of the metric tensor can always be written as the Euler-Lagrange expression of some Lagrangian function. Moreover, by the solution of the equivariant inverse problem of the calculus of variations for metric field theories presented in \cite{And84}, the Lagrangian can be chosen to be natural when the dimension of the underlying space is $m=0,1,2 \mod 4$, and, in dimensions $m=3\mod4$, the system can be written as a sum of the Euler-Lagrange expression of a natural Lagrangian and a generalized Cotton tensor.

Extending our result to fourth order operators either by showing that natural, divergence free systems are necessarily locally variational or by finding non-variational examples of such systems remains a challenging open problem. 
In contrast, for vector field theories \cite{AnPo95} and, in more generality, for Yang-Mills theories \cite{MPV08}, non-variational third order source forms with the prescribed symmetries and conservation laws can be derived by the way of natural constructions utilizing the intrinsic geometry of the problem.
Note that unlike in the situation of the present paper, proving that a 4th order source form $\Tsf$ fulfills the highest order non-vacuous \He conditions $\Helmcomp abcd{i_1i_2i_3i_4}=0$ will not be sufficient to guarantee that $\Tsf$ is locally variational.  

A significant generalization of the present work would be to enlarge the symmetry
group to include conformal transformations of the metric with the resultant condition that the source form be trace free. The problem closely bears on conformal gravity, and, as exemplified by the Bach equations \cite{Bach21}, will require the analysis of fourth order systems. A solution for the equivariant inverse problem of the calculus of variations in this situation would also be of substantial independent interest. Additional on physical grounds intriguing problems would be Takens' question for combined field theories involving metric, such as Einstein-Yang-Mills equations in both the abelian and non-abelian cases.

\end{document}